\documentclass{article}

\usepackage{arxiv}

\usepackage[utf8]{inputenc} 
\usepackage[T1]{fontenc}    
\usepackage{hyperref}       
\usepackage{url}            
\usepackage{booktabs}       
\usepackage{amsfonts}       
\usepackage{nicefrac}       
\usepackage{microtype}      
\usepackage{lipsum}
\usepackage{graphicx}
\usepackage[UKenglish]{babel}
\usepackage{comment}
\usepackage{cite}
\usepackage{appendix}
\usepackage{amssymb}

\usepackage{proof}
\usepackage{graphicx}
\usepackage{float}
\usepackage{color}
\usepackage{hyperref}

\usepackage[all]{xy}
\usepackage{tikz} 
\usetikzlibrary{trees}
\usepgflibrary{arrows}
\usetikzlibrary{positioning,shapes}
\usetikzlibrary{patterns}

\newtheorem{definition}{Definition}
\newtheorem{lemma}{Lemma}
\newtheorem{theorem}{Theorem}
\newtheorem{proof}{Proof}
\newtheorem{remark}{Remark}

\newtheorem{proposition}{Proposition}
\renewcommand{\phi}{\varphi}
\newcommand{\M}{\mathcal{M}}
\newcommand{\N}{\ensuremath{\mathcal{N}}}
\newcommand{\Kw}{\textit{Kw}}
\newcommand{\Ew}{\textit{Ew}}
\newcommand{\Cw}{\textit{Cw}}
\newcommand{\C}{\textit{C}}

\newcommand{\G}{\textsf{G}}
\newcommand{\BP}{\textsf{P}}
\newcommand{\K}{\ensuremath {\textit{K}}}
\newcommand{\E}{\ensuremath {\textit{E}}}

\newcommand{\PLKwCw}{\textbf{Cw}}
\newcommand{\PLKC}{\ensuremath{\textbf{C}}}
\newcommand{\lr}[1]{\langle #1 \rangle}
\newcommand{\lra}{\leftrightarrow}
\newcommand{\weg}[1]{}

\title{Commonly Knowing Whether}

\author{
  Jie Fan\thanks{Jie Fan and Xingchi Su are main authors of this paper.}, Davide Grossi, Barteld Kooi, Xingchi Su$^\star$, Rineke Verbrugge}

\begin{document}
\maketitle

\begin{abstract}
This paper introduces the notion of `commonly knowing whether', a non-standard version of standard common knowledge which is defined on the basis of `knowing whether', instead of standard `knowing that'. After giving five possible definitions of this notion, we explore the logical relations among them in the single-agent and multi-agent cases. We propose a sound and complete axiomatization. We investigate one of the five definitions in terms of expressivity via a strategy of modal comparison games.
\end{abstract}

\section{Introduction}

Common knowledge, or `commonly knowing that', as the strongest concept among group epistemic notions, has been studied extensively in various areas such as artificial intelligence, epistemic logic, epistemology, philosophy of language, epistemic game theory, see e.g.~\cite{parikh1992levels,meyer2004epistemic,fagin2004reasoning,Lewis1969Convention,Clark1981Definite,aumann1995epistemic}. Intuitively, a proposition is common knowledge among a group of agents, if the proposition is true, everyone (in the group) knows it, everyone knows everyone knows it, everyone knows everyone knows everyone knows it, and so on ad infinitum.\footnote{This so-called `iterate approach', attributed to~\cite{Lewis1969Convention}, is the most common and orthodox view of common knowledge. There are also approaches called `the fixed-point approach' and `the shared-environment approach', see~\cite{barwise2016three} for the comparisons among the three approaches and also for related references.} common knowledge is defined based on the notion of `knowing that'.

Beyond `knowing that', recent years have witnessed a growing interest in other types of knowledge, such as `knowing whether', `knowing what', `knowing how', `knowing why', `knowing who', see~\cite{Wang:2018} for an excellent overview. Among these notions, `knowing whether' is the closest friend of `knowing that'.     The notion of `Knowing whether' is used frequently to specify knowledge goals and preconditions for actions, see e.g.~\cite{McCarthy79,Reiter01,petrick2004extending}. Besides, `knowing whether' corresponds to an important notion of philosophy, that is, non-contingency, the negation of contingency, which dates back to Aristotle~\cite{Brogan67}. An agent {\em knows whether} a proposition $\phi$ holds, if the agents knows that $\phi$ is true, or the agent knows that $\phi$ is false; otherwise, the agent is ignorant about $\phi$. A proposition $\phi$ is {\em non-contingent}, if it is necessary that $\phi$, or it is impossible that $\phi$; otherwise, $\phi$ is contingent. Just as `knowing that' is the epistemic counterpart of necessity, `knowing whether' is the epistemic counterpart of non-contingency. For an overview of contingency and `knowing whether', we recommend~\cite{Fan2015CONTINGENCY}.

It is therefore natural to ask what the notion of `commonly knowing whether' is in contrast to the notion of common knowledge based on `knowing that'. `Commonly knowing whether' is not equivalent to common knowledge. For instance, suppose you see two people chatting beside a window but you cannot look outside yourself. Then you know that they {\em commonly know whether} it is sunny outside, but you do not know that they {\em commonly know that} it is sunny outside since you do not see the weather. There has been no unanimous agreement yet on the formal definition of `commonly knowing whether'. As we will show, there are at least five definitions for this notion, which are not logically equivalent over various frame classes.

There has been no unanimous agreement yet on the formal definition of `commonly knowing whether'. As we will show, there are at least five definitions for this notion, which are not logically equivalent over various frame classes. Also, we will demonstrate that one of the definitions is {\em not} expressible with common knowledge. 

The remainder of the paper is organized as follows. Section~\ref{sec.definitions} gives five definitions for the notion of `commonly knowing whether', which are based on distinct intuitions. Section~\ref{sec.implicationrelations} compares their implicational powers in the single-agent and also multi-agent cases. Section~\ref{sec.axiomatization} gives an axiomatization and demonstrates its soundness and completeness. Section~\ref{sec.expressivity} explores the relative expressivity of one of the five definitions, via a strategy of modal comparison games. Moreover, we explore some special properties of $Cw_5$ over binary trees in Section~\ref{sec.binarytrees}. Finally, we conclude with some future work in Section~\ref{sec.conclusions}.

\section{Definitions of `Commonly Knowing Whether' ($Cw$)}\label{sec.definitions}

This section presents some definitions for the notion of `commonly knowing whether'. To begin with, we fix a denumerable set of propositional variables $\BP$ and a nonempty finite set of agents $\G$. We use $\G^+$ to refer to the set of finite and nonempty sequences only consisting of agents from $\G$, and $|\G|$ to refer to the number of agents in $\G$. The languages involved in this paper are defined recursively as follows:
\[
\begin{array}{ll}
\PLKwCw&\phi::= p \mid \neg \phi \mid (\phi\land\phi)\mid \Kw_i\phi\mid \Ew\phi \mid \Cw\phi\\
\PLKC&\phi::= p \mid \neg \phi \mid (\phi\land\phi)\mid \K_i\phi\mid \E\phi \mid \C\phi\\

\end{array}
\]
where $p\in \BP$, and $i\in \G$. In this paper, for the sake of simplicity, we only concern $\G$ but not any of its proper subsets. 

Intuitively, $\Kw_i\phi$ is read ``agent $i$ knows whether $\phi$'', $\Ew\phi$ is read ``everyone knows whether $\phi$'', $\Cw\phi$ is read ``a set of agents commonly knows whether $\phi$'', and $\K_i, \E,\C$ are more familiar operators of individual knowledge, general knowledge, and common knowledge, respectively.

As noted in the introduction, according to the iterate approach, common knowledge is defined as infinite iteration of general knowledge (namely `everyone knows'). Similarly, we can give an iterate approach to commonly knowing whether, according to which commonly knowing whether amounts to infinite iteration of `everyone knows whether'. For this, we need to define the notion of `everyone knowing whether'. One definition for `everyone knowing whether' is similar to the notion of `everyone knows', that is, to say everyone in a group knows whether $\phi$, if every agent in the group knows whether $\phi$~(Def.~\ref{def.ew2}). This is seemingly the most natural notion of `everyone knowing whether'. Another definition for `everyone knowing whether' is similar to (individual) knowing whether; namely, everyone in a group knows whether $\phi$, if everyone in the group knows that $\phi$ is true, or everyone in the group knows that $\phi$ is false~(Def.~\ref{def.ew1}).

\begin{definition}\label{def.ew1}
$\Ew_1\phi:=\E\phi\vee \E\neg\phi$
\end{definition}

\begin{definition}\label{def.ew2}
$\Ew_2\phi:=\bigwedge_{i\in\G} \Kw_i\phi$
\end{definition}

It should be easy to check that $\Ew_1$ is stronger than $\Ew_2$. This can be informally explained as follows: for any agent, if (s)he knows $\phi$ or knows $\neg\phi$, then (s)he knows whether $\phi$ is true; however, the other direction fails, because it is possible that some agents know that $\phi$ is true but others know that $\phi$ is false.

Based on the above definitions of `everyone knowing whether' and common knowledge, we propose the following possible definitions for `commonly knowing whether'.

\begin{definition}\label{def.cw1}
$\Cw_1\phi:= \C\phi\vee \C\neg\phi$
\end{definition}

The definition of $\Cw_1$ is structurally similar to that of $\Ew_1$. Intuitively, it says that a group {\em commonly knows whether} $\phi$, if the group has either common knowledge of $\phi$ or has common knowledge of the negation of $\phi$. $\Cw_1$ corresponds to the notion of `commonly knowing whether' in the sunny weather example in the introduction. Besides, this definition may also find applications in question-answer contexts. suppose that Sue is attending a conference and asking the speaker a question: ``Is $q$ true or false?'' No matter whether the speaker says `Yes' or `No', all attendees will {\em commonly know whether} $p$. Since if the speaker says `Yes', then the attendees will commonly know that $p$ is true; otherwise, the attendees will commonly know that $p$ is false. The answer of the speaker amounts to an announcement whether $p$~\cite{van2016propositional} (depending on the truth value of $p$), which leads to the common knowledge as to the truth value of $p$.

\begin{definition}\label{def.cw2}
$\Cw_2\phi:= \C\Ew\phi$
\end{definition}

According to this definition, a group {\em commonly knows whether} $\phi$, if it is common knowledge that everyone knows whether $\phi$. Since there are two different definitions of $\Ew$, we have also two different definitions of $\Cw_2$, that is, $\Cw_{21}\phi:=\C\Ew_1\phi$ and $\Cw_{22}\phi:=\C\Ew_2\phi$.

\begin{definition}\label{def.cw3}
$\Cw_3\phi:= \bigwedge_{k\geq 1}(\Ew)^k\phi$
\end{definition}

This is the iterated approach to commonly knowing whether: a group {\em commonly knows whether} $\phi$, if everyone (in the group) knows whether $\phi$, everyone knows whether everyone knows whether $\phi$, everyone knows whether everyone knows whether everyone knows whether $\phi$, and so on ad infinitum. $\Cw_{3}$ corresponds to the notion of `commonly knowing whether' in the muddy children example in the introduction. As indicated in that example, {\em neither} $q$ {\em nor} $\neg q$ is common knowledge among the single-agent set of the mud child, thus $\Cw_{3}$ is not stronger than $\Cw_1$. Again, because there are two different definitions of $\Ew$, we have also two different definitions of $\Cw_3$, that is, $\Cw_{31}\phi:=\bigwedge_{k\geq 1}(\Ew_1)^k\phi$ and $\Cw_{32}\phi:=\bigwedge_{k\geq 1}(\Ew_2)^k\phi$.

\begin{definition}\label{def.cw4}
$\Cw_4\phi:= \bigwedge_{i\in \G}\Cw_1\Kw_i\phi$, that is, $\Cw_4\phi:= \bigwedge_{i\in \G}(\C\Kw_i\phi\vee \C\neg \Kw_i\phi)$
\end{definition}

According to this definition, a group {\em commonly knows whether} $\phi$, if for every member (in the group), it is common knowledge that (s)he knows whether $\phi$ or it is common knowledge that (s)he does not know whether $\phi$.

\begin{definition}\label{def.cw5}
$\Cw_5\phi := \bigwedge_{s\in \G^+}\Kw_s\phi$, where $\Kw_s\phi:= \Kw_{s_1} \dots \Kw_{s_n}\phi$ if $s=s_1\dots s_n$ is a nonempty sequence of agents.
\end{definition}

This definition is inspired by the hierarchy of inter-knowledge of a group given in \cite{parikh1992levels}. According to this definition, `commonly knowing whether' amounts to listing all the possible inter-`knowing whether' states among every nonempty subset of the group.
.

\section{Implication Relations among the definitions of $Cw$}\label{sec.implicationrelations}

This section explores the implication relations among the above five definitions of `commonly knowing whether'. The semantics of $\PLKwCw$ and $\PLKC$ are interpreted over Kripke models. A {\em (Kripke) model} is a tuple $\M=\lr{W,\{R_i\mid i\in\G\},V}$, where $W$ is a nonempty set of (possible) worlds (also called `states'), $R_i$ is an accessibility relation for each agent $i$, and $V$ is a valuation function assigning a set of possible worlds to each propositional variable in $\BP$. We say that $(\M,s)$ is a {\em pointed model}, if $\M$ is a model and $s$ is a world in $\M$. We use $\to$ to denote the union of all $R_i$ for $i\in\G$, and $\twoheadrightarrow$ to denote the reflexive-transitive closure of $\to$, that is, $\twoheadrightarrow=\bigcup_{n\in\mathbb{N}}\to^n$. A {\em (Kripke) frame} is a model without valuations. Moreover, we use $\mathcal{K}$, $\mathcal{T}$, $\mathcal{KD}45$, and $\mathcal{S}5$ to denote the class of all Kripke frames, the class of reflexive frames, the class of serial, transitive and Euclidean frames, and the class of reflexive and Euclidean frames, respectively.

Given a model $\M=\lr{W,\{R_i\mid i\in\G\},V}$ and a world $w\in W$, the semantics of $\PLKwCw$ and $\PLKC$ are defined inductively as follows (we do not list the semantics of $\Ew$ and $\Cw$ due to their alternative definitions).

\begin{definition}
The truth conditions for the formulas in $\PLKwCw$ and $\PLKC$ are shown as follows:
\begin{itemize}
\item $\M,w\vDash p \Leftrightarrow  w\in V(p)$

\item $\M,w\vDash\neg\phi\Leftrightarrow  \M,w\nvDash\phi$

\item $\M,w\vDash\phi\land\psi\Leftrightarrow \M,w\vDash\phi\ and\ \M,w\vDash\psi$

\item $\M,w\vDash\Kw_i\phi\Leftrightarrow for\ all\ u,v\in W,\ 
if\ wR_iu\ and\ wR_iv,\ 
then\ (\M,u\vDash\phi\iff \M,v\vDash\phi)$

\item $\M,w\vDash\K_i\phi\Leftrightarrow for\ all\ u\in W, if\ wR_iu,\
then\ \M,u\vDash\phi$

\item $\M,w\vDash\E\phi\Leftrightarrow for\ all\ u\in W,\ if\ w\to u,\ then\ \M,u\vDash\phi$

\item $\M,w\vDash\C\phi\Leftrightarrow for\ all\ u\in W,\  if\ w\twoheadrightarrow u,\ 
then\ \M,u\vDash\phi$
\end{itemize}
\end{definition}

Notions of satisfiability and validity are defined as normal. For instance, we say that $\phi$ is {\em valid over the class of frames $\mathcal{C}$}, written $\mathcal{C}\vDash\phi$, if for all frames $\mathcal{F}$ in $\mathcal{C}$, for all models $\M$ based on $\mathcal{F}$, for all worlds $w$ in $\M$, we have $\M,s\vDash\phi$; we say that $\phi$ is {\em valid}, if $\phi$ is valid over the class of all frames $\mathcal{K}$.

The following result will simplify the later proofs. We omit the proof details due to the space limitation.

\begin{proposition}
\begin{enumerate}
    \item $\vDash \C\phi\to \C\E\phi$
    \item $\vDash \phi\to\psi$ implies $\vDash \C\phi\to \C\psi$.
\end{enumerate}
\end{proposition}

\subsection{Single-agent Case}

This part investigates the implication relations among definitions of `commonly knowing whether'. First, we explore the implication relations between the two definitions of `everyone knowing whether'. It turns out that $\Ew_1$ and $\Ew_2$ are equivalent in the single-agent case.

\begin{proposition}\label{prop.single-agent-equiv}
If $|\G|=1$, then for any $\phi$, $\vDash\Ew_1\phi\lra\Ew_2\phi$.
\end{proposition}

\begin{proof}
Let $|\G|=1$, say $\G=\{i\}$. In this case, $\Ew_2\phi$ is just $\Kw_i\phi$. Also, $\vDash \E\phi\lra \K_i\phi$, and $\vDash\E\neg\phi\lra \K_i\neg\phi$, and thus $\vDash\E\phi\vee\E\neg\phi\lra \K_i\phi\vee\K_i\neg\phi$, that is,  $\vDash\Ew_1\phi\lra\Ew_2\phi$.
\end{proof}

This tells us that in the single-agent case, we do not need to distinguish between $\Cw_{21}$ and $\Cw_{22}$, and between $\Cw_{31}$ and $\Cw_{32}$, and thus we only talk about $\Cw_2$ and $\Cw_3$, respectively, when only one agent is involved. In the reminder of this subsection, we assume $\G$ to be a single-agent set $\{i\}$.

\begin{proposition}\label{prop.Cw1toCw2}
$\vDash \Cw_1\phi\to\Cw_{2}\phi$, $\vDash \Cw_{2}\phi\to\Cw_{3}\phi$, and therefore, $\vDash \Cw_1\phi\to\Cw_{3}\phi$.
\end{proposition}

\begin{proof}
Since $\vDash \C\phi\to\C\E\phi$, and $\vDash\C\E\phi\to\C\Ew\phi$ (because $\vDash\E\phi\to\Ew\phi$), we have $\vDash\C\phi\to\C\Ew\phi$. Similarly, we obtain $\vDash\C\neg\phi\to\C\Ew\neg\phi$. It is easy to see that $\vDash\Ew\phi\lra\Ew\neg\phi$. Therefore, $\vDash\C\phi\vee\C\neg\phi\to\C\Ew\phi$, i.e., $\vDash\Cw_1\phi\to\Cw_2\phi$.

Moreover, since $\C\Ew\phi=\Ew\phi\land\E\Ew\phi\land\cdots$, and $\vDash \E\phi\to\Ew\phi$, we can show that $\vDash\C\Ew\phi\to\Ew\phi\land\Ew\Ew\phi\land\cdots$, that is, $\vDash\Cw_2\phi\to\Cw_3\phi$.
\end{proof}

\begin{proposition}\label{prop.Cw2ntoCw1}
$\nvDash\Cw_{2}\phi\to\Cw_1\phi$.
\end{proposition}

\begin{proof}
Consider the following model $\M$:
\[
\xymatrix{s:\neg p\ar[rr]|i&&t: p}
\]
Clearly, $\M,s\vDash\K_i p$ and $\M,t\vDash\K_i p$. Then $\M,s\vDash\E p$ and $\M,t\vDash \E p$. We can obtain that $\M,s\vDash\C\Ew p$. However, $\M,s\nvDash\C p\vee\C\neg p$: Since $\M,s\nvDash  p$, we have $\M,s\nvDash\C p$; since $\M,t\nvDash\neg p$, we have $\M,s\nvDash\C\neg p$. Therefore, $\nvDash\Cw_{2} p\to\Cw_1p$.
\end{proof}

\begin{proposition}\label{prop.cw3ntocw2}
$\nvDash\Cw_{3}\phi\to\Cw_{2}\phi$, and thus $\nvDash\Cw_3\phi\to\Cw_1\phi$.
\end{proposition}

\begin{proof}
Consider the following model $\mathcal{N}$:
\[
\xymatrix{s:p\ar[rr]|i&&t:\neg p\ar[ll]|i\ar@(ur,ul)|i}
\]
Since $s$ has only one successor, $\mathcal{N},s\vDash\Kw_i\phi$ for all $\phi$, thus $\mathcal{N},s\vDash\Ew\phi$ for all $\phi$, and hence $\mathcal{N},s\vDash \Cw_3\phi$. In particular, $\mathcal{N},s\vDash\Cw_3p$.

However, $\mathcal{N},t\nvDash\Kw_ip$, thus $\mathcal{N},t\nvDash\Ew p$, and then $\mathcal{N},s\nvDash\C\Ew p$, that is, $\mathcal{N},s\nvDash\Cw_2p$. Therefore, $\nvDash\Cw_{3}p\to\Cw_{2}p$.
\end{proof}

\begin{proposition}\label{prop.cw2tocw4}
$\vDash \Cw_2\phi\to \Cw_4\phi$. As a corollary, $\vDash\Cw_1\phi\to\Cw_4\phi$.
\end{proposition}

\begin{proof}
Since $\vDash\C\Ew\phi\lra\C\Kw_i\phi$, whereas $\vDash\Cw_1\Kw_i\phi\lra
(\C\Kw_i\phi\vee\C\neg\Kw_i\phi)$, therefore $\vDash\C\Ew\phi\to(\C\Kw_i\phi\vee\C\neg\Kw_i\phi)$, that is, $\vDash \Cw_2\phi\to \Cw_4\phi$.
\end{proof}

\begin{proposition}\label{prop.cw3ntocw4}
$\nvDash\Cw_3\phi\to\Cw_4\phi$.
\end{proposition}

\begin{proof}
Consider the model $\mathcal{N}$ in the proof of Prop.~\ref{prop.cw3ntocw2} again.
On the one hand, we have shown that $\mathcal{N},s\vDash\Cw_3 p$. On the other hand, since $\mathcal{N},t\nvDash\Kw_i p$, we have that $\mathcal{N},s\nvDash\C\Kw_ip$; moreover, because $\mathcal{N},s\nvDash\neg\Kw_ip$, we infer that $\mathcal{N},s\nvDash\C\neg\Kw_ip$, which implies that $\mathcal{N},s\nvDash\Cw_4p$. Therefore, $\nvDash\Cw_3p\to \Cw_4p$.
\end{proof}

\begin{proposition}\label{prop.cw4ntocw3}
$\nvDash\Cw_4\phi\to\Cw_3\phi$. As a corollary, $\nvDash\Cw_4\phi\to\Cw_2\phi$ and $\nvDash\Cw_4\phi\to\Cw_1\phi$.
\end{proposition}

\begin{proof}
Consider the following model:
\[
\xymatrix{s:p\ar@(ul,ur)|i\ar[rr]|i&&t:\neg p\ar@(ur,ul)|i\ar[ll]|i}
\]
It is easy to check that $s\vDash\neg\Kw_i p$ and $t\vDash\neg\Kw_i p$, from which we can obtain $s\vDash\C\neg\Kw_i p$. Then $s\vDash\Cw_4 p$.

However, since $s\nvDash\Kw_ip$, we have $s\nvDash\Ew p$, i.e. $s\nvDash\Cw_3 p$. Therefore, $\nvDash \Cw_4p\to \Cw_3p$. 
\end{proof}

From the proofs of the Props.~\ref{prop.cw2tocw4} and~\ref{prop.cw4ntocw3}, it follows that even on $\mathcal{S}5$-models, $\Cw_4$ is not logically equivalent to $\Cw_1$, $\Cw_2$ and $\Cw_3$. 

\begin{proposition}
$\vDash\Cw_3\phi\lra \Cw_5\phi$.
\end{proposition}

\begin{proof}
Since in the case that $\G=\{i\}$, $\vDash\Ew\phi\lra\Kw_i\phi$, we have $\vDash \Cw_3\phi\lra \bigwedge_{k\geq 1}(\Kw_i)^k\phi$. Moreover, $\vDash \Cw_5\phi\lra (\Kw_i\phi\land\Kw_i\Kw_i\phi\land\cdots)$. So, $\vDash\Cw_3\phi\lra \Cw_5\phi$.
\end{proof}

Summarize the main results of this subsection as follows.

\begin{theorem}\label{th.12345k}
In the single-agent case, the logical relationships among $\Cw_1$, $\Cw_2$, $\Cw_3$, $\Cw_4$ and $\Cw_5$ are shown in Fig.~\ref{fig.single}, where an arrow from one operator $O$ to another $O'$ means that for all $\phi$, $O\phi\to O'\phi$ is valid over the class of frames in question, and the arrows are transitive.
\end{theorem}

\begin{figure}
\[
\xymatrix@R-10pt{\Cw_1\ar[r]&\Cw_2\ar[r]\ar[d]&\Cw_3\ar[d]\\
&\Cw_4&\Cw_5\ar[u]\\}
\]
\caption{Logical relationships in the single-agent case}~\label{fig.single}
\end{figure}
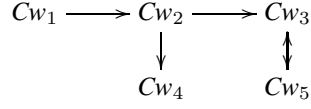

\subsection{Multi-agent Cases}

Now we move on to the multi-agent cases, that is, the cases when $|\G|>1$. As explained before, it is possible that some agents knows that $\phi$ is true but others know that $\phi$ is false, thus $\Ew_1$ is stronger than $\Ew_2$.

\begin{proposition}\label{prop.ew1toew2}
The following statements hold:
\begin{itemize}
\item[(a)]$\vDash \Ew_1\phi\to \Ew_2\phi$; 
\item[(b)]$\nvDash \Ew_2\phi\to \Ew_1\phi$.
\end{itemize}
\end{proposition}

\begin{proof}
(a) is straightforward by definitions of $\Ew_1$ and $\Ew_2$ and the semantics. For (b), consider the following model $\M$: 
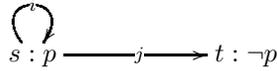
\[\xymatrix{s:p\ar@(ul,ur)|{i}\ar[rr]|j&& t:\neg p}\]

It is clear that $\M,s\vDash \Kw_ip\land\Kw_jp$, thus $\M,s\vDash\Ew_2p$. However, $\M,s\nvDash\E p$ and $\M,s\nvDash\E \neg p$, thus $\M,s\nvDash\Ew_1p$. Therefore, $\nvDash\Ew_2 p\to\Ew_1p$.
\end{proof}

Due to Prop.~\ref{prop.ew1toew2}, in the multi-agent cases, we need to distinguish between $\Ew_1$ and $\Ew_2$, and thus need to distinguish between $\Cw_{21}$ and $\Cw_{22}$, and also between $\Cw_{31}$ and $\Cw_{32}$.

The following is the main result in this part.

\begin{theorem}\label{th.12345km}

In the multi-agent cases, the logical relationships among $\Cw_1$, $\Cw_{21}$, $\Cw_{22}$, $Cw_{31}$, $\Cw_{32}$, $\Cw_4$ and $\Cw_5$ are shown in Fig.~\ref{fig.one}.
\end{theorem}

\begin{figure}
\[
\xymatrix@R-10pt{  &&& \Cw_5\\ \Cw_1 \ar[r] & \Cw_{21}\ar[r]\ar[dr]
& \Cw_{22} \ar[ur] \ar[dr] \ar[r] & \Cw_4 \\ && \Cw_{31}  & \Cw_{32}}
\]

\caption{Logical relationships in the multi-agent cases}~\label{fig.one}
\end{figure}
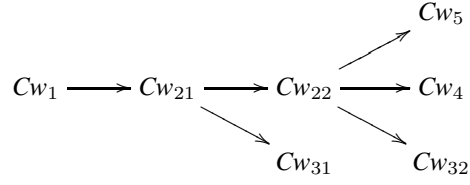

\begin{proposition}\label{prop.Cw-implies}
$\vDash \Cw_1\phi\to\Cw_{21}\phi$, $\vDash\Cw_{21}\phi\to\Cw_{22}\phi$. Consequently, $\vDash \Cw_1\phi\to\Cw_{22}\phi$.
\end{proposition}

\begin{proof}
The first part is immediate from Prop.~\ref{prop.Cw1toCw2}.

Moreover, as $\vDash \Ew_1\phi\to\Ew_2\phi$ (by Prop.~\ref{prop.ew1toew2} (a)), thus $\vDash \C\Ew_1\phi\to\C\Ew_2\phi$, that is, $\vDash \Cw_{21}\phi\to\Cw_{22}\phi$.
\end{proof}

\begin{proposition}
$\vDash \Cw_{21}\phi\to\Cw_{31}\phi$, $\vDash \Cw_{22}\phi\to\Cw_{32}\phi$. 
As a consequence, $\vDash\Cw_{21}\phi\to\Cw_{32}\phi$.
\end{proposition}

\begin{proof}
Similar to the proof in Prop.~\ref{prop.Cw1toCw2}.

\end{proof}

\begin{proposition}
$\nvDash\Cw_{22}\phi\to\Cw_{31}\phi$. As a consequence, $\nvDash\Cw_{22}\phi\to\Cw_{21}\phi$.
\end{proposition}

\begin{proof}

Consider the model $\M$ in the proof of Prop.~\ref{prop.ew1toew2} (b).

We have seen that $\M,s\vDash\Ew_2p$. Moreover, it is obvious that $\M,t\vDash\Ew_2p$, and thus $\M,s\vDash \C\Ew_2p$. However, $\M,s\nvDash\E p$ and $\M,s\nvDash\E \neg p$, thus $\M,s\nvDash\Ew_1p$, and hence $\M,s\nvDash\Cw_{31}p$. Therefore, $\nvDash\Cw_{22}p\to\Cw_{31}p$.
\end{proof}

\begin{proposition}\label{prop.cw31ntocw32}
$\nvDash\Cw_{31}\phi\to\Cw_{32}\phi$. As a result, $\nvDash\Cw_{31}\phi\to\Cw_{22}\phi$.
\end{proposition}

\begin{proof}
Consider the following model:
\[
\xymatrix{t_3: p&&t_4:\neg p\\
t_1:p\ar[u]|i\ar@/^12pt/[urr]|i&s:p\ar[l]|i\ar[r]|i&t_2:p\ar[u]|i\ar@/_12pt/[ull]|j\\}
\]

Firstly, $s\vDash\Cw_{31}p$. To see this, notice that $s\vDash\Ew_1p$. Since $\M,t_1\nvDash\Ew_1p$ and $\M,t_2\nvDash\Ew_1p$, we infer that $\M,s\vDash\Ew_1\Ew_1p$. Moreover, as $t_3$ and $t_4$ both have no successors, $t_3\vDash\Ew_1\phi$ and $t_4\vDash\Ew_1\phi$ for all $\phi$, then $t_1\vDash\Ew_1\Ew_1\phi$ and $t_2\vDash\Ew_1\Ew_1\phi$, and therefore $s\vDash\Ew_1\Ew_1\Ew_1\phi$ for all $\phi$. This implies that $s\vDash\Cw_{31}p$.

Secondly, $s\nvDash\Cw_{32}p$. To see this, note that $t_1\nvDash\Ew_2p$ and $t_2\vDash\Ew_2p$. This implies that $s\nvDash\Ew_2\Ew_2p$, and therefore $s\nvDash\Cw_{32}p$. Now we conclude that $\nvDash\Cw_{31}p\to\Cw_{32}p$.
\end{proof}

\begin{proposition}\label{prop.cw32ntocw31}
$\nvDash\Cw_{32}\phi\to\Cw_{31}\phi$.
\end{proposition}

\begin{proof}
Use the model $\M$ in the proof of Prop.~\ref{prop.ew1toew2} (b). We have seen there that $\M,s\nvDash\Ew_1p$, and thus $\M,s\nvDash\Cw_{31}p$.

Moreover, as $s$ has only one $i$-successor and only one $j$-successor, we have $\M,s\vDash\Kw_i\phi\land\Kw_j\phi$ for all $\phi$, that is, $\M,s\vDash\Ew_2\phi$ for all $\phi$, and thus $\M,s\vDash\Cw_{32}\phi$ for all $\phi$, hence $\M,s\vDash\Cw_{32}p$. Therefore, $\nvDash\Cw_{32}p\to \Cw_{31}p$.
\end{proof}

\begin{proposition}\label{prop.notequal}
$\nvDash\Cw_{21}\phi\to\Cw_1\phi$, and thus $\nvDash\Cw_{22}\phi\to\Cw_1\phi$.
\end{proposition}

\begin{proof}
Similar to the proof of Prop.~\ref{prop.Cw2ntoCw1}.
\weg{For the first one, consider the following model $\M$, where $i$ is a sole agent:
\[
\xymatrix{s:\neg p\ar[rr]|i&&t: p}
\]
Clearly, $\M,s\vDash\K_i p$ and $\M,t\vDash\K_i p$. Then $\M,s\vDash\E p$ and $\M,t\vDash \E p$. We can obtain that $\M,s\vDash\C\Ew_1 p$. However, $\M,s\nvDash\C p\vee\C\neg p$: Since $\M,s\nvDash  p$, we have $\M,s\nvDash\C p$; since $\M,t\nvDash\neg p$, we have $\M,s\nvDash\C\neg p$. Therefore, $\nvDash\Cw_{21} p\to\Cw_1p$.}
\end{proof}

\begin{proposition}
$\nvDash\Cw_{31}\phi\to\Cw_{21}\phi$, and $\nvDash\Cw_{32}\phi\to\Cw_{22}\phi$.
\end{proposition}

\begin{proof}
Similar to the proof of Prop.~\ref{prop.cw3ntocw2}.
\weg{Consider the following model $\mathcal{N}$, where $i$ is a sole agent:
\[
\xymatrix{s:p\ar[rr]|i&&t:\neg p\ar[ll]|i\ar@(ur,ul)|i}
\]
Since $s$ has only one successor, $s\vDash\Kw_i\phi$ for all $\phi$, thus $s\vDash\Ew\phi$ for all $\phi$, and hence $s\vDash \Cw_3\phi$. In particular, $s\vDash\Cw_3p$.

However, $t\nvDash\Kw_ip$, thus $t\nvDash\Ew p$, and then $s\nvDash\K_i\Ew p$,  and hence $s\nvDash\E\Ew p$. Therefore $s\nvDash\Cw_2p$.}
\end{proof}

\begin{proposition}
$\vDash\Cw_{22}\phi\to \Cw_4\phi$, and thus $\vDash\Cw_{21}\phi\to\Cw_4\phi$.
\end{proposition}

\begin{proof}
By definition of $\Ew_2$, we have $\vDash\Ew_2\phi\lra \bigwedge_{i\in\G}\Kw_i\phi$, then $\vDash\C\Ew_2\phi\lra \C\bigwedge_{i\in\G}\Kw_i\phi$, that is, $\vDash \C\Ew_2\phi\lra \bigwedge_{i\in\G}\C\Kw_i\phi$, and thus $\vDash \Cw_{22}\phi\to \Cw_4\phi$.
\end{proof}

\begin{proposition}
$\nvDash\Cw_4\phi\to \Cw_{32}\phi$ and $\nvDash\Cw_4\phi\to \Cw_{31}\phi$. Consequently, $\nvDash\Cw_4\phi\to \Cw_{22}\phi$, $\nvDash\Cw_4\phi\to\Cw_{21}\phi$, $\nvDash\Cw_4\phi\to \Cw_1\phi$.
\end{proposition}

\begin{proof}
Consider the model in the proof of Prop.~\ref{prop.cw4ntocw3}. We have already seen that $s\vDash\Cw_4p$.

However, since $s\nvDash\Kw_ip$, thus $s\nvDash\Ew_2p$, which implies that $s\nvDash\Ew_1p$ due to Prop.~\ref{prop.ew1toew2} (a). It follows that $s\nvDash\Cw_{32}p$ and $s\nvDash\Cw_{31}p$. Therefore, $\nvDash\Cw_4p\to\Cw_{32}p$ and $\nvDash\Cw_4p\to\Cw_{31}p$.
\end{proof}

\begin{proposition}
$\vDash \Cw_{22}\phi\to \Cw_5\phi$.
\end{proposition}

\begin{proof}
Let $\M=\lr{W,\{R_i\mid i\in\G\},V}$ be a model and $w\in W$. Suppose that $\M,w\vDash\Cw_{22}\phi$, to show that $\M,w\vDash\Cw_5\phi$, that is, $\M,w\vDash\bigwedge_{s\in G^+}\Kw_s\phi$. Let $s\in G^+$ be arbitrary. It suffices to prove that $\M,w\vDash\Kw_s\phi$. For this, we show a stronger result: $(\ast)$ for any $u$ such that $w\twoheadrightarrow u$, we have that $\M,u\vDash Kw_s\phi$. From this it follows immediately that $\M,w\vDash\Kw_s\phi$ due to the fact that $w\twoheadrightarrow w$.

We proceed by induction on the length of $s$, denoted by $|s|$. By supposition, for all $u$ such that $w\twoheadrightarrow u$, we have $\M,u\vDash\bigwedge_{i\in \G}\Kw_i\phi$.

Base step: $|s|=1$. We may assume that $s=i_1$, where $i_1\in\G$. For any $u$ with $w\twoheadrightarrow u$, as $\M,u\vDash\bigwedge_{i\in \G}\Kw_i\phi$, we have obviously that $\M,u\vDash \Kw_{i_1}\phi$. Therefore, $(\ast)$ holds.

Induction step: hypothesize that $(\ast)$ holds for $|s|=k$ (IH), we prove that $(\ast)$ also holds for $|s|=k+1$. For this, we may assume that $s=i_1i_2\cdots i_{k+1}$, where $i_m\in\G$ for $m\in[1,k+1]$.

For all $u$ such that $w\twoheadrightarrow u$, for all $v$ such that $uR_{i_1}v$, we have $w\twoheadrightarrow v$. Note that $|i_2\cdots i_{k+1}|=k$. Then by IH, we derive that $\M,v\vDash\Kw_{i_2\cdots i_{k+1}}\phi$, and therefore, $\M,u\vDash\Kw_{i_1i_2\cdots i_{k+1}}\phi$, that is, $\M,u\vDash\Kw_s\phi$. We have now shown $(\ast)$, as desired.
\end{proof}

\begin{proposition}\label{prop.cw32ntocw5}
$\nvDash\Cw_{32}\phi\to\Cw_5\phi$. As a consequence, $\nvDash\Cw_{32}\phi\to\Cw_{22}\phi$.
\end{proposition}

\begin{proof}
Consider the following model:
\[
\xymatrix{
s:p\ar[rr]|i\ar[drr]|i&&t:p\ar[rr]|{i,j}\ar@/^5pt/[drr]|j&&v:p\\
&&u:p\ar[rr]|{i,j}\ar@/_5pt/[urr]|i&&w:\neg p\\}
\]

In this model, on one hand, $s\vDash\Cw_{32}p$: note that $s\vDash\Ew_2p$ is easily verified; since $t\nvDash\Kw_jp$, thus $t\nvDash\Ew_2 p$, similarly, since $u\nvDash\Kw_ip$, thus $u\nvDash\Ew_2p$, from which it follows that $s\vDash\Ew_2\Ew_2p$. Moreover, since both $v$ and $w$ have no successors, $\Ew_2 p\land\Ew_2\Ew_2 p\land\cdots$ holds at both $v$ and $w$, then $\Ew_2\Ew_2p\land\Ew_2\Ew_2\Ew_2p\land\cdots$ holds at both $t$ and $u$. Therefore, $s\vDash\Ew_2\Ew_2\Ew_2p\land\Ew_2\Ew_2\Ew_2\Ew_2p\land\cdots$, as desired.

On the other hand, $t\vDash\Kw_ip$ but $u\nvDash\Kw_ip$, and hence $s\nvDash\Kw_i\Kw_ip$. This entails that $s\nvDash\Cw_5p$, as desired.
\end{proof}

\begin{proposition}\label{prop.cw5ntocw32}
$\nvDash\Cw_5\phi\to\Cw_{32}\phi$.
\end{proposition}

\begin{proof}
Consider the following model:
\[
\xymatrix{v_1:p&v_2:\neg p&&v_3:p&v_4:\neg p\\
u_1:p\ar[u]|i\ar@/_10pt/[ur]|i&u_2:p\ar[u]|j\ar@/^10pt/[ul]|j&&u_3:p&u_4:p\ar[ul]|{i,j}\ar[u]|{i,j}\\
&t_1:p\ar[ul]|i\ar[u]|i&&t_2:p\ar[u]|i\ar[ur]|i&&\\
&&s:p\ar[ul]|i\ar[ur]|i&&\\}
\]

We now show that $\M,s\vDash\Cw_5p$ but $\M,s\nvDash\Cw_{32}p$.

One may easily see that $\M,s\vDash\Kw_j\phi$ for all $\phi$, and $s\vDash \Kw_ip$, thus $s\vDash\Kw_ip\land\Kw_jp$. Moreover, $t_1$ and $t_2$ both satisfy $\Kw_j\phi$ for all $\phi$ and $\Kw_ip$, and then $s\vDash\Kw_i\Kw_ip\land\Kw_i\Kw_jp$. Besides, $u_1\vDash\neg\Kw_ip\land\Kw_jp$ and $u_2\vDash\Kw_ip\land\neg\Kw_jp$, which implies that $t_1\vDash\neg\Kw_i\Kw_ip\land\neg\Kw_i\Kw_jp$; $u_3\vDash\Kw_i\phi\land\Kw_j\phi$ for any $\phi$, $u_4\vDash\neg\Kw_ip\land\neg\Kw_jp$, thus $t_2\vDash\neg\Kw_i\Kw_ip\land\neg\Kw_i\Kw_jp$, and hence $s\vDash\Kw_i\Kw_i\Kw_ip\land\Kw_i\Kw_i\Kw_jp$. Of course $s\vDash\Kw_i\Kw_j\phi$ for all $\phi$, thus $s\vDash\Kw_i\Kw_j\Kw_ip\land\Kw_i\Kw_j\Kw_jp$. Because $v_1,v_2,v_3,v_4$ all satisfy $\Kw_a\phi$ for all $\phi$ and $a$, we can obtain that $s\vDash\Kw_a\Kw_b\Kw_c\Kw_d\phi$ for all $a,b,c,d\in\G$. Therefore, $\M,s\vDash\Cw_5p$.

Since $u_k\nvDash\Kw_ip\land\Kw_jp$ for $k=1,2$, thus $t_1\vDash\Kw_i\Ew_2 p$, and then $t_1\vDash\Kw_i\Ew_2 p\land\Kw_j\Ew_2p$, i.e. $t_1\vDash\Ew_2\Ew_2p$. However, $u_3\vDash\Ew_2p$ and $u_4\nvDash\Ew_2p$, which entails that $t_2\nvDash\Kw_1\Ew_2p$, and thus $t_2\nvDash\Ew_2\Ew_2p$. Therefore, $s\nvDash\Kw_i\Ew_2\Ew_2p$, and thus $s\nvDash\Ew_2\Ew_2\Ew_2p$. This leads to $\M,s\nvDash\Cw_{32}p$.
\end{proof}

\begin{proposition}\label{prop.cw31ntocw5}
$\nvDash\Cw_{31}\phi\to\Cw_5\phi$.
\end{proposition}

\begin{proof}
Consider the model in the proof of Prop.~\ref{prop.cw32ntocw5}. It has been shown there that $s\nvDash\Cw_5p$. It then suffices to prove that $s\vDash\Cw_{31}p$.

As $s\vDash \K_ip$ and $s\vDash\K_j\phi$ for all $\phi$, we obtain $s\vDash \E p$, and thus $s\vDash\Ew_1p$. Since $t\nvDash\E p\vee\E\neg p$ and $u\nvDash \E p\vee\E\neg p$, we can show that $s\vDash\E\neg \Ew_1 p$, and then $s\vDash\Ew_1\Ew_1 p$. Moreover, because $v$ and $w$ both have no successors, $v\vDash\Ew_1\phi$ and $w\vDash\Ew_1\phi$ for all $\phi$, thus $t\vDash\Ew_1\Ew_1\phi$ and $u\vDash\Ew_1\Ew_1\phi$, and hence $s\vDash\Ew_1\Ew_1\Ew_1\phi$ for all $\phi$. This together gives us $s\vDash\Cw_{31}p$.
\end{proof}

\begin{proposition}
$\nvDash\Cw_5\phi\to\Cw_{31}\phi$.
\end{proposition}

\begin{proof}
Consider the model in the proof of Prop.~\ref{prop.cw5ntocw32}. It has been shown there that $s\vDash\Cw_5p$. The remainder is to show that $s\nvDash\Cw_{31}p$. The proof is as follows.

One may check that $\Ew_1p$ is true at $u_3$ but false at $u_1,u_2,u_4$. It then follows that $t_1\vDash\Ew_1\Ew_1p$ but $t_2\nvDash\Ew_1\Ew_1p$, and therefore $s\nvDash\Ew_1\Ew_1\Ew_1p$, which implies that $s\nvDash\Cw_{31}p$, as desired.
\end{proof}

\begin{proposition}\label{prop.cw32ntocw4}
$\nvDash\Cw_{32}\phi\to\Cw_4\phi$.
\end{proposition}

\begin{proof}
Use the model in the proof of Prop.~\ref{prop.cw32ntocw5}. It has been shown there that $s\vDash\Cw_{32}p$. 

Moreover, $s\nvDash\Cw_4p$. To see this, note that $t\nvDash\neg\Kw_ip$ and $u\nvDash\Kw_ip$. This entails that $s\nvDash\C\Kw_ip\vee\C\neg\Kw_ip$. Therefore, $s\nvDash\Cw_4p$.
\end{proof}

\begin{proposition}\label{prop.cw31ntocw4}
$\nvDash\Cw_{31}\phi\to\Cw_4\phi$.
\end{proposition}

\begin{proof}
Use the model in the proof of Prop.~\ref{prop.cw32ntocw5}. By Prop.~\ref{prop.cw32ntocw4}, we have $s\nvDash\Cw_4p$. By Prop.~\ref{prop.cw31ntocw5}, we obtain $s\vDash\Cw_{31}p$. Therefore, $\nvDash\Cw_{31}p\to\Cw_4p$.
\end{proof}

\begin{proposition}\label{th.n45k}
$\nvDash \Cw_4\phi\to \Cw_5\phi$.
\end{proposition}

\begin{proof}
Consider the model in the proof of Prop.~\ref{prop.cw4ntocw3}.
We have seen that $M,s\models Cw_4p$. However, as $M,s\vDash \neg Kw_ip$, we can obtain that $M,s\nvDash \Cw_5p$. So, $\nvDash \Cw_4 p\to \Cw_5p$.
\end{proof}

\begin{proposition}\label{prop.cw5ntocw4}
$\nvDash\Cw_5\phi\to\Cw_4\phi$.
\end{proposition}

\begin{proof}
Use the model in the proof of Prop.~\ref{prop.cw5ntocw32}. There, we have shown that $\M,s\vDash\Cw_5p$. It is now sufficient to prove that $\M,s\nvDash\Cw_4p$.

Note that $u_1\nvDash\Kw_ip$ and $u_3\nvDash\neg\Kw_ip$. This gives us $s\nvDash\C\Kw_ip\vee\C\neg\Kw_ip$, and therefore $s\nvDash\Cw_4p$.
\end{proof}

\weg{\begin{proposition}\label{th.n45k}
$\nvDash \Cw_4\phi\to \Cw_5\phi$.
\end{proposition}

\begin{proof}
Consider the following model:
\[
\xymatrix@R-10pt{
t_3:p\ar[r]|i\ar@(ur,ul)|i&t_4:\neg p\ar[l]|i\ar@(ur,ul)|i&&t_5:p\ar[r]|i\ar@(ur,ul)|i&t_6:\neg p\ar[l]|i\ar@(ur,ul)|i\\
&t_1:p\ar[ul]|i\ar[u]|i&&t_2:\neg p\ar[u]|i\ar[ur]|i&&\\
&&s:p\ar[ul]|i\ar[ur]|i&&\\}
\]

In this model, $Kw_ip$ is false on every world, thus $M,s\models Cw_4p$. However, as $M,s\models \neg Kw_ip$, we can obtain that $M,s\not\models \Cw_5p$. Therefore, $\nvDash \Cw_4 p\to \Cw_5p$.
\end{proof}}

\weg{\subsection{Single-agent Case}
The results above concern the multi-agent case. Considering the  single-agent case, observe that $\mathcal{K}\models E_1w\phi\leftrightarrow E_2w\phi$, which implies that $C_{21}w$ is equivalent to $C_{22}w$ and that $C_{31}w$ is equivalent to $C_{32}w$. Moreover, since only one agent is involved, it is trivial to find that $C_3w$ is equivalent to $C_5w$. 

In terms of logical relationships, the five definitions share the same relations with the multi-agent case over $\mathcal{K}$, $\mathcal{KD}45$, $\mathcal{T}$ and $\mathcal{S}5$, respectively.} 

We have thus completed the proof of Thm.~\ref{th.12345k}. Note that all proofs involved are based on $\mathcal{K}$. In fact, we can also explore the implication relations 
of $\Cw_1-\Cw_5$ over $\mathcal{KD}45$, over $\mathcal{T}$, and also over $\mathcal{S}5$. It turns out that the implication relations over $\mathcal{KD}45$ is the same as those over $\mathcal{K}$ (see Fig.~\ref{fig.one}), and the implication relations over $\mathcal{T}$ and $\mathcal{S}5$ is figured as follows (see Fig.~\ref{fig.s5}). We omit the proof details due to the space limitation.

\begin{figure}
\[
\scriptsize
\xymatrix@C-10pt{\Cw_1 \ar[r] \ar[d] &\Cw_{21} \ar[r]  \ar[l] &\Cw_{22} \ar[l] \ar[r] &\Cw_{31} \ar[r] \ar[l] &\Cw_{32} \ar[r] \ar[l] &\Cw_5 \ar[l] \\ \Cw_4 }
\]
\caption{Logical relationships over $\mathcal{T}$ and $\mathcal{S}5$ (multi-agent).}\label{fig.s5}
\end{figure}
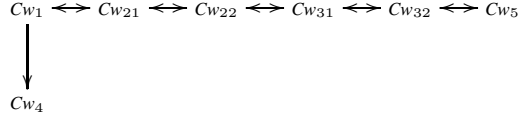

Therefore, $\Cw_1$, $\Cw_2$, $\Cw_3$ and $\Cw_5$ boil down to the same thing once the frame is reflexive, which can be attributed to agents' agreements on the values of $\phi$. For example, if there is $\M,s\models \Kw_i\phi\wedge \Kw_j\phi$ where $\M$ is reflexive, the values of $\phi$ on all $i$-successors must agree with those on all $j$-successors since they share a common successor $s$. Comparatively, if $\M$ is not reflexive, the case of $\M,s\models \K_i\phi\wedge \K_j\neg \phi$ is possible.

\section{Axiomatization}\label{sec.axiomatization}

$\mathcal{S}5$ is the class of frames specifically for knowledge or epistemic description. As mentioned above, over $\mathcal{S}5$, the five definitions of `commonly knowing whether' are logically equivalent except $\Cw_4$. Also, one may verify that $\Ew_1$ and $\Ew_2$ are logically equivalent over $\mathcal{S}5$. In this section, we axiomatize $\PLKwCw$ over $\mathcal{S}5$. The language $\PLKwCw$ can now be defined recursively as follows:
\[
\begin{array}{l}
\phi::= p \mid \neg \phi \mid (\phi\land\phi)\mid \Kw_i\phi \mid \Cw\phi,\\
\end{array}
\]
where $\Cw$ means $\Cw_1$, and $\Ew\phi$ abbreviates $\bigwedge_{i\in\G}\Kw_i\phi$.

The semantics of $\PLKwCw$ is defined as before in addition to the semantics of $\Cw$ as follows.
\begin{definition}
   $\M,w\vDash \Cw\phi  \Leftrightarrow for\ all\  u,v,\ if\ w\twoheadrightarrow u\ and\ w\twoheadrightarrow v,\ then\  (\M,u\vDash\phi\iff \M,v\vDash\phi).$
\end{definition}

\subsection{Proof system and Soundness}

The proof system $\mathbb{PLCKW}5$ is an extension of the axiom system of the logic of `knowing whether' $\mathbb{CLS}5$ in \cite{Fan2015CONTINGENCY} plus the axioms and rules concerning $\Cw$.

\begin{definition}
The axiomatization of $\mathbb{PLCKWS}5$ consists of the following axiom schemas and inference rules:

\begin{center}
 \begin{tabular}{c c c c} 
 
All instances of tautologies & (TAUT) \\

$\Kw_i(\chi\to\phi)\wedge \Kw_i(\neg\chi\to\phi)\to \Kw_i\phi$ & (Kw-CON)\\

$\Kw_i\phi\to \Kw_i(\phi\to\psi)\vee \Kw_i(\neg\phi\to\chi)$ & (Kw-DIS)\\

$\Kw_i\phi\wedge \Kw_i(\phi\to \psi)\wedge \phi\to \Kw_i\psi$ & (Kw-T)\\

$\neg \Kw_i\phi\to\Kw_i\neg \Kw_i\phi$ & (wKw-5)\\

$\Kw_i\varphi \leftrightarrow \Kw_i \neg\varphi$ & (Kw-$\leftrightarrow$)\\

$\Cw(\chi\to\phi)\wedge \Cw(\neg\chi\to\phi)\to \Cw\phi$ & (Cw-CON)\\

$\Cw\phi\to \Cw(\phi\to\psi)\vee \Cw(\neg\phi\to\chi)$ & (Cw-DIS)\\

$\Cw\phi\wedge \Cw(\phi\to \psi)\wedge \phi\to \Cw\psi$ & (Cw-T)\\

$\Cw\phi\to (\Ew\phi\wedge \Ew\Cw\phi)$ & (Cw-Mix)\\

$\Cw(\phi\to \Ew\phi)\to (\phi\to \Cw\phi)$ & (Cw-Ind)\\

from $\phi$  infer $\Kw_i\phi$ & (Kw-NEC)\\

from $\phi$ infer $\Cw\phi$ & (Cw-NEC)\\

from $\phi\leftrightarrow\psi$ infer $ \Kw_i\phi\leftrightarrow \Kw_i\psi$ & (Kw-RE)\\

from $\phi\leftrightarrow\psi$ infer $\Cw\phi\leftrightarrow \Cw\psi$ & (Cw-RE)\\

from $\phi$ and $\phi\to\psi$ infer $\psi$ & (MP)
\end{tabular}
\end{center}
\end{definition}



\begin{proposition}
$\mathbb{PLCKWS}5$ is sound with respect to $\mathcal{S}5$.
\end{proposition}

   \begin{proof}
The soundness of the axioms (Kw-CON), (Kw-DIS), (Kw-T), (wKw-5) and the soundness of the rules (Kw-NEC), (Kw-RE), (MP) are already proved in \cite{Fan2015CONTINGENCY}. Moreover, the soundness of (Cw-CON), (Cw-DIS), (Cw-T), (Cw-NEC) and (Cw-RE) can be shown as their Kw-counterparts. It suffices to show the soundness of (Cw-Mix) and (Cw-Ind).  Let $\M=\lr{W,\{R_i\mid i\in\G\},V}$ be an arbitrary $\mathcal{S}5$-model and $w\in W$.

\medskip

For (Cw-Mix): Suppose that $\M,w\vDash\Cw\phi$, to show that $\M,w\vDash \Ew\phi\land\Ew\Cw\phi$. We only show that $\M,w\vDash\Ew\Cw\phi$, since $\M,w\vDash\Ew\phi$ is straightforward. 

If $\M,w\nvDash\Ew\Cw\phi$, then for some $i\in\G$ and some $u,v$ such that $wR_iu$ and $wR_iv$ we have $\M,u\vDash\Cw\phi$ and $\M,v\nvDash\Cw\phi$. Then there exist $v_1,v_2$ such that $v\twoheadrightarrow v_1$ and $v\twoheadrightarrow v_2$ and $\M,v_1\vDash\phi$ and $\M,v_2\nvDash\phi$. Then $w\twoheadrightarrow v_1$ and $w\twoheadrightarrow v_2$, and therefore $\M,w\nvDash\Cw\phi$, which is contrary to the supposition.

\medskip

For (Cw-Ind):  Suppose that $\M,w\vDash \Cw(\phi\to \Ew\phi)\wedge\phi$. We show a stronger result: $(\ast)$: for all $u$ such that $w\twoheadrightarrow u$, we have $\M,u\vDash\phi$; equivalently, for all $n\in\mathbb{N}$, for all $u$ such that $w\to^nu$, we have $\M,u\vDash\phi$ (Recall that $\twoheadrightarrow=\bigcup_{n\in\mathbb{N}}\to^n$). From $\M,w\vDash\Cw(\phi\to\Ew\phi)$, it follows that {\em either} (1) for all $u$ such that $w\twoheadrightarrow u$ we have $\M,u\vDash\phi\to\Ew\phi$ {\em or} (2) for all $u$ such that $w\twoheadrightarrow u$ we have $\M,u\vDash\phi\land\neg\Ew\phi$. The case (2) immediately gives us $(\ast)$. In the remainder, it suffices to consider the case (1).

We proceed with induction on $n$.

Base step (i.e. $n=0$). In this case, we need to show that $\M,w\vDash\phi$. This follows directly from the supposition.

Inductive step. Assume by induction hypothesis (IH) that $(\ast)$ holds for $n=k$. We show that $(\ast)$ also holds for $n=k+1$. Hypothesize that $w\to^{k+1}u$, then there exists $v$ such that $w\to^kv\to u$. By IH, $\M,v\vDash\phi$. Obviously, $w\twoheadrightarrow v$, then using (1) we infer that $\M,v\vDash\phi\to\Ew\phi$, thus $\M,v\vDash\Ew\phi$. Together with the fact that $\M,v\vDash\phi$ and the reflexivity of $\to$ (since $R_i$ is reflexive for all $i\in\G$), this would implies that for all $x$ such that $v\to x$ we have $\M,x\vDash\phi$, and therefore $\M,u\vDash\phi$, as desired.
\end{proof}

\subsection{Completeness of $\mathbb{PLCKWS}5$}

We follow the basic idea on proving completeness of the logic of common knowledge to prove the completeness of $\mathbb{PLCKWS}5$.

\begin{definition}
The {\em closure} of $\phi$, denoted as $cl(\phi)$, is the smallest set satisfying following conditions:
\begin{enumerate}
    \item $\phi\in cl(\phi)$.
    \item if $\psi\in cl(\phi)$, then $sub(\psi)\subseteq cl(\phi)$.
    \item if $\psi\in cl(\phi)$ and $\psi$ is not itself of the form $\neg\chi$, then $\neg\psi\in cl(\phi)$.
    \item if $\Kw_i\psi\in cl(\phi)$ and $\Kw_i\chi\in cl(\phi)$, and $\chi$ and $\psi$ are not themselves  conditionals, then $\Kw_i(\chi\to\psi)\in cl(\phi)$.
    \item if $\Cw\psi\in cl(\phi)$, then $\{\Kw_i\Cw\psi \mid i\in \G\}\subseteq cl(\phi)$.
    \item If $\Cw\psi\in cl(\phi)$, then $\{\Kw_i\psi\mid i\in\G\}\subseteq cl(\phi)$.
    \item If $\neg \Kw_i\psi_1\in cl(\phi)$, $\psi_1$ is not a negation and $\Cw\psi_2\in cl(\phi)$, then $\Kw_i\neg(\psi_1\wedge\neg\psi_2)\in cl(\phi)$.
\end{enumerate}
\end{definition}

The definition of the canonical model for $\mathbb{PLCKWS}5$ is inspired by the canonical models of $\mathbb{CLS}5$ in \cite{Fan2015CONTINGENCY}, where the definition of the  canonical relation is inspired by an almost schema.

\begin{definition}
$\Phi$ is a closure of some formula. We define {\em the canonical model} based on $\Phi$ as $\M^c=\langle W^c, \{R_i^T\mid i\in\G\}, V^c \rangle$ where:

\begin{enumerate}
    \item $W^c=\{\Sigma\mid\Sigma\ is\ maximal\ consistent\ in\ \Phi\}$.
    \item For each $i\in \G$, let $R_i^T$  be the reflexive closure of $R_i^c$, where $\Sigma R_i^c \Delta$ iff there exists an $\chi$ which is not a conditionals, such that: 
          \begin{enumerate}
              \item $\neg \Kw_i\chi\in \Sigma$ and 
              \item for all $\phi\in\Phi$: ($\Kw_i\phi\in\Sigma$ and $\Kw_i\neg(\phi\wedge\neg\chi)\in \Sigma$) implies $\phi\in\Delta$.
          \end{enumerate}
          
    \item $V^c(p)=\{\Sigma\in W^c\mid p\in\Sigma\}$.
\end{enumerate}
\end{definition}

Here, it should be noticed that it has already been proved that $R_i^T$ is an equivalence relation in \cite{Fan2015CONTINGENCY}, which means we no longer need the  transitive and symmetric closure. A useful proposition should also be given in advance.

\begin{proposition}\label{prop.bigwedge}
For all $n>1$:
\begin{enumerate}
    \item $\vdash (Kw_i(\bigwedge_{k=1}^n\phi_k\to\neg\psi)\wedge \bigwedge_{k=1}^n Kw_i\phi_k\wedge \bigwedge_{k=1}^n Kw_i(\psi\to\phi_k))\to Kw_i\psi$
    \item $\vdash (Kw_i(\bigwedge_{k=1}^n\phi_k\to\psi)\wedge \bigwedge_{k=1}^n Kw_i\phi_k\wedge \bigwedge_{k=1}^n Kw_i(\neg\psi\to\phi_k))\to Kw_i\psi$
\end{enumerate}
\end{proposition}

The proof of Proposition \ref{prop.bigwedge}.1 is given in \cite{Fan2015CONTINGENCY} and Proposition \ref{prop.bigwedge}.2 can thereby be directly derived.

\begin{lemma}\label{le.lindenbaum}
(Lindenbaum's Lemma) Let $\Phi$ be the closure of some formula. Every consistent subset of $\Phi$ is a subset of a maximal consistent set in $\Phi$.
\end{lemma}

The proof of Lemma \ref{le.lindenbaum} is standard. In order to prove the completeness of $\mathbb{PLCKWS}5$, we refer to the basic idea of the completeness of classical common knowledge in \cite{van2007dynamic} and \cite{fagin2004reasoning} and Lemma \ref{le.Cwtruthlemma} should be proved in advance.

\begin{definition}\label{def.path}
($\alpha$-path) In the canonical model $M^c$, if $\Sigma\twoheadrightarrow\Delta$, then the sequence of maximal-consistent sets $l=\langle \Gamma_0,\Gamma_1,\cdots,\Gamma_n \rangle$ satisfying following two conditions:\footnote{Here we abuse the notation and use $\twoheadrightarrow$ to denote the reflexive-transitive closure of the union of the canonical relations.}
\begin{enumerate}
    \item $\Sigma=\Gamma_0$, $\Delta=\Gamma_n$;
    \item for any $k$ ($0\leq k<n$), there is an agent $i\in G$ such that $\Gamma_k R^T_i\Gamma_{k+1}$;
    \item for any $m$ ($0\leq m \leq n$), $\alpha\in\Gamma_m$.
\end{enumerate}

is an {\em $\alpha$-path} from $\Sigma$ to $\Delta$.
\end{definition}

\begin{lemma}\label{le.Cwtruthlemma}
If $\Cw\phi\in\Phi$, then $\Cw\phi\in\Sigma$ iff every path from $\Sigma$ is a $\phi$-path or every path from $\Sigma$ is a $\neg\phi$-path.
\end{lemma}

\begin{proof}
($\Rightarrow$) Proof by induction on the length $n$ of the path. We have to prove a stronger lemma: if $\Cw\phi\in\Sigma$, then every path from $\Sigma$ is a $\phi$-path and a $\Cw\phi$-path or every path from $\Sigma$ is a $\neg\phi$-path and $\Cw$-path.

When $n=0$, $l=\langle\Sigma\rangle$. Since $\Sigma$ is maximal consistent and $\phi\in\Phi$, $\phi\in\Sigma$ or $\neg\phi\in\Sigma$. And $\Cw\phi\in\Sigma$ is our premise.
 
Induction hypothesis: If $\Cw\phi\in\Sigma$, then every path of length $n$ is a $\phi$-path and a $\Cw\phi$-path or every path of length $n$ is a $\neg\phi$-path and a $\Cw\phi$-path.

Induction Step: Assume $l_n=\langle \Sigma_0,\Sigma_1,...,\Sigma_n \rangle$ is a $\phi$-path and $\Cw\phi$-path. Take a path of length $(n+1)$ from $\Sigma$. By induction hypothesis, $\Cw\phi\in\Sigma_n$. Let $i$ be the agent such that $\Sigma_n R_i^T \Sigma_{n+1}$. Since $\mathbb{PLCKWS}5\vdash \Cw\phi\to \Kw_i\phi$ and $\Kw_i\phi\in\Phi$, it must be the case that $\Kw_i\phi\in \Sigma_n$. Suppose $\phi\not\in \Sigma_{n+1}$. By the definition of $R_i^T$, there exists a $\chi$ which is not conditionals such that $\neg\Kw_i\phi\in\Sigma_n$, and $\Kw_i\phi\not\in\Sigma_n$ or $\Kw_i(\chi\in\phi)$. By $\Kw_i\phi\in\Sigma_n$, we can infer that $\Kw_i(\chi\to\phi)\not\in\Sigma_n$. Since $\Cw\phi\in\Phi$, we know that $\phi$ is itself not a conditionals. So $\Kw_i(\chi\to\phi)\in\Phi$. Thus, $\neg\Kw_i(\chi\to\phi)\in\Sigma_n$. By (Kw-Dis) and (Kw-T), we have $\{\Kw_i\phi,\neg\Kw_i(\chi\to\phi),\phi\}\vdash\Kw_i\chi$, which implies that $\Kw_i\chi\in\Sigma_n$. It contradicts to $\neg\Kw_i\chi\in\Sigma_n$. Therefore, $\phi\in\Sigma_{n+1}$.

Since $\mathbb{PLCKWS}5\vdash \Cw\phi\to \Ew\Cw\phi$ and $\mathbb{PLCKWS}5\vdash \Ew\Cw\phi\to \Kw_i\Cw\phi$, it must be the case that $\Kw_i\Cw\phi\in \Sigma_n$. Similarly, we can prove that $\Cw\phi\in\Sigma_{n+1}$. Thus, for any path with (n+1) length, it is a  $\phi$-path and $\Cw\phi$-path.

Assume $l_n=\langle \Sigma_0,\Sigma_1,...,\Sigma_n \rangle$ is a $\neg\phi$-path and $\Cw\phi$-path. We can also prove that for any path with (n+1) length, it is a $\neg\phi$-path and $\Cw\phi$-path.

($\Leftarrow$) Suppose that every path from $\Sigma$ is a $\phi$-path or every path from $\Sigma$ is a $\neg\phi$-path. We want to prove that $\vdash\bigwedge\Sigma\to \Ew\phi$. Assume, to reach a contradiction, that $\bigwedge\Sigma\wedge\neg\Ew\phi$ is consistent. It implies that there exists $i\in\G$ such that $\bigwedge\Sigma\wedge\neg\Kw_i\phi$ is consistent. By $\Cw\phi\in\Phi$, $\neg\Kw_i\phi\in\Phi$. So $\neg\Kw_i\phi\in\Sigma$. And we know that $\phi$ is not a conditionals. Now we are to construct two maximal consistent sets in $\Phi$, $\Gamma_1$ and $\Gamma_2$ such that $\Sigma R_i^T\Gamma_1$ and $\Sigma R_i^T\Gamma_2$ and $\phi\in\Gamma_1$ and $\neg\phi\in\Gamma_2$. Firstly, we will show the following two items: 

\begin{enumerate}
    \item\label{enumerate.1} $\{\theta\mid \theta\ \rm{is\ not\ conditionals\ and}\ \Kw_i\theta\in\Sigma\ \rm{and}\ \Kw_i(\phi\to\theta)\in\Sigma \}\cup\{\phi\}$ is consistent.
    \item\label{enumerate.2} $\{\theta\mid \theta\ \rm{is\ not\ conditionals\ and}\ \Kw_i\theta\in\Sigma\ \rm{and}\ \Kw_i(\neg\phi\to\theta)\in\Sigma \}\cup\{\neg\phi\}$ is consistent.
\end{enumerate}

As for~\ref{enumerate.1}, assume that it is not consistent. It implies that there exist $\theta_1,\cdots,\theta_n$ in it such that $\vdash(\theta_1\wedge\cdots\wedge\theta_n)\to\neg\phi$ and $\Kw_i\theta_k\in\Sigma$ and $\Kw_i(\phi\to\theta_k)\in\Sigma$ for all $1\leq k\leq n$. By (Kw-NEC), $\vdash\Kw_i((\theta_1\wedge\cdots\theta_n)\to\neg\phi)$. By Proposition \ref{prop.bigwedge}.1, we infer that $\Kw_i\phi\in\Sigma$. Contradiction.

As for~\ref{enumerate.2}, assume that it is not consistent. It implies that there exist $\theta_1,\cdots,\theta_n$ in it such that $\vdash(\theta_1\wedge\cdots\wedge\theta_n)\to\phi$ and $\Kw_i\theta_k\in\Sigma$ and $\Kw_i(\neg\phi\to\theta_k)\in\Sigma$ for all $1\leq k\leq n$. By (Kw-NEC), $\vdash\Kw_i((\theta_1\wedge\cdots\theta_n)\to\phi)$. By Proposition \ref{prop.bigwedge}.2, we infer that $\Kw_i\phi\in\Sigma$. Contradiction.

Thus, these two consistent sets can be extended to two maximal consistent sets $\Gamma_1$ and $\Gamma_2$, according to Lindenbaum lemma. Accroding to the definition of $R_i^T$, $\Sigma R_i^T\Gamma_1$ and $\Sigma R_i^T\Gamma_2$ and $\phi\in\Gamma_1$ and $\neg\phi\in\Gamma_2$. It contradicts to our supposition that every path from $\Sigma$ is a $\phi$-path or every path from $\Sigma$ is a $\neg\phi$-path. So we have $\vdash\bigwedge\Sigma\to\Ew\phi$. By (Cw-NEC) and (Cw-Ind), $\vdash\bigwedge\Sigma\to\Cw\phi$. Therefore, $\Cw\phi\in\Sigma$.
\end{proof}

\begin{lemma}\label{le.truthlemma}
(Finite Truth Lemma) For any $\PLKwCw$-formula $\psi$, for all $\Sigma\in W^c$, we have $\M^c,\Sigma\models\psi$ iff $\psi\in\Sigma$.
\end{lemma}

\begin{proof}
By induction on $\psi$:

$\bullet$ When $\psi$ is a Boolean formula or is $\Kw_i\phi$, the proof  can be shown as in \cite{Fan2015CONTINGENCY}.
    
$\bullet$ When $\psi=\Cw\phi$:
    
    `Only if': Suppose that $\Cw\phi\not\in\Sigma$. By Lemma $\ref{le.Cwtruthlemma}$, there exist two paths $l_1$ and $l_2$ such that $l_1$ is not a $\phi$-path and $l_2$ is not a $\neg\phi$-path. Thus, there must be a $\Delta_1\in l_1$ such that $\neg\phi\in\Delta_1$ and a $\Delta_2\in l_2$ such that $\phi\in\Delta_2$. By induction hypothesis, we have $M^c,\Delta_1\vDash\neg\phi$ and $M^c,\Delta_2\vDash\phi$. By the definition \ref{def.path}, $\Sigma\twoheadrightarrow\Delta_1$ and $\Sigma\twoheadrightarrow\Delta_2$. Thus, $M^c,\Sigma
    \nvDash \Cw\phi$.
    
    `If': Assume $\M^c,\Sigma\not\models \Cw\phi$. By the semantics of $\Cw\phi$, there exist $\Delta_1$, $\Delta_2\in W^c$ with $\Sigma\twoheadrightarrow\Delta_1$, $\Sigma\twoheadrightarrow\Delta_2$, such that $\M^c,\Delta_1\models\phi$ and $\M^c,\Delta_2\models \neg\phi$. By induction hypothesis, $\phi\in\Delta_1$ and $\neg\phi\in\Delta_2$. Thus, there exists a path $l_1$ where $\Delta_1\in l_1$ such that $l_1$ us not a $\neg\phi$-path and there also exists a path $l_2$ where $\Delta_2\in l_2$ such that $l_2$ is not a $\phi$-path. By Lemma \ref{le.Cwtruthlemma}, $\Cw\phi\not\in\Sigma$.
\end{proof}

By Lemma \ref{le.truthlemma}, we obtain the completeness of $\mathbb{PLCKWS}5$.

\begin{theorem}\label{th.completeness}
The logic $\mathbb{PLCKWS}5$ is weakly complete with respect to $\mathcal{S}5$.
\end{theorem}

\section{Expressivity}\label{sec.expressivity}

In this section, we will compare the expressivity of $\Cw_5$ with that of common knowledge, since both notions are inspired by the hierarchy of inter-knowledge of a group given in \cite{parikh1992levels}. The two languages are:

\[
\begin{array}{l l}
\PLKwCw_5 &  \phi ::= p \mid \neg \phi \mid (\phi \land \phi)\mid \Kw_i\phi \mid \Cw_5\phi \\  \PLKC &  \phi::= p\mid \neg\phi\mid (\phi\wedge\phi)\mid K_i\phi \mid C\phi \nonumber
\end{array}
\]

\subsection{$\PLKwCw_5$ is Bisimulation Invariant}

\begin{definition}
Let $M=\langle W,R,V\rangle$ and $M'=\langle W',R',V'\rangle$ be two Kripke models. A non-empty binary relation $Z\subseteq W\times W'$ is called bisimulation between $M$ and $M'$, written as $M\cong M'$, if the following conditions are satisfied:

(i) If $wZ'w'$, then $w$ and $w'$ satisfy the same proposition letters.

(ii) if $wZ'w'$ and $wRv$, then there is a $v'\in W'$ such that $vZv'$ and $w'R'v'$.

(iii) If $wZ'w'$ and $w'R'v'$, then there exists $v\in W$ such that $vZv'$ and $wRv$.

When $Z$ is a bisimulation linking two states $w$ in $M$ and $w'$ in $M'$, we say that two pointed models are bisimilar and write $Z: (M,w)\cong(M',w')$. If a language $L$ cannot distinguish any pair of bisimilar models, $L$ is bisimulation invariant.
\end{definition}

\begin{theorem}
$\PLKwCw_5$ is bisimulation invariant.
\end{theorem}

\begin{proof}
By induction on formulas $\phi$ of $\PLKwCw_5$.

When $\phi$ is a Boolean formula, the proof is classical.

When $\phi=Kw_i\psi$, we prove it in three cases. For arbitrary two bisimilar models $(M,r)$ and $(N,t)$, we have:

\begin{itemize}
    \item if $M,r\models Kw_i\psi$ and for all $r_n$ with $r\to_Mr_n$, $M,r_n\models\psi$. Since $M,r\cong N,t$, for any $t_n$ with $t\to_Nt_n$, there is an $r_n$ such that $r\to_Mr_n$ and $M,r_n\cong N,t_n$. By induction hypothesis, $\psi$ is bisimulation invariant. Thus $N,t_n\models\psi$. So $N,t\models K_i\psi$.
    
    \item if $M,r\models Kw_i\psi$ and for all $r_n$ with $r\to_Mr_n$, $M,r_n\models\neg\psi$, similar to above case.
    
    \item if $M,r\models\neg Kw_i\psi$, that means there are $r_1$ with $r\to_Mr_1$ and $r_2$ with $r\to_Mr_2$, such that $M,r_1\models\psi$ and $M,r_2\models\neg\psi$. Since $M,r\cong N,t$, there are $t_1$ with $t\to_N t_1$ and $t_2$ with $t\to_N t_2$, such that $M,r_1\cong N,t_1$ and $M,r_2\cong N,t_2$. By induction hypothesis, $\psi$ is bisimulation invariant. Thus $N,t_1\models\psi$ and $N,t_2\models\neg\psi$. Thus $N,t\models \neg Kw_i\psi$.
    
Thus, $Kw_i\psi$ is bisimulation invariant.    
\end{itemize}

When $\phi=Cw_5\psi$, assume two bisimular models $(M,r)$ and $(N,t)$, such that $M,r\models Cw_5\psi$ and $N,t\models \neg Cw_5\psi$. That means there exists a sequence of agents $s$, such that $M,r\models Kw_s\psi$ and $N,t\models \neg Kw_s\psi$. Let $s=\langle i_1,i_2,...,i_n \rangle$. So $M,r\models Kw_{i_1}\gamma$ and $N,t\models\neg Kw_{i_1}\gamma$, where $\gamma=Kw_{\langle i_2,i_3,...,i_n\rangle}\psi$. However, we have proved that for any formula of the form $Kw_{i_1}\gamma$, they are bisimulation invariant. Thus, if $M,r\models Kw_{i_1}\gamma$, there must be $N,t\models Kw_{i_1}\gamma$. Contradiction.

\smallskip

Therefore, we proved that $\PLKwCw_5$ is bisimulation invariant.
\end{proof}

\subsection{$\PLKC$ and $\PLKwCw_5$}

Although $\Cw_5$ is formed merely with $\Kw$, which can be defined by classical operator $\K$, surprisingly, $\Cw_5$ is not expressible in $\PLKC$. We prove it by constructing two classes of models, which {\em no} $\PLKC$-formula can distinguish whereas some $\PLKwCw_5$-formula can. The following definitions and lemmas facilitate our proofs.

\begin{definition}{(Modal Depth)}
The modal depth of a $\PLKC$-formula is defined by:
\[
\begin{array}{ll}
    d(p)=1; & d(\neg\phi)=d(\phi);\\ d(\phi\wedge\psi)=max\{d(\phi),d(\psi)\};& d(\K\phi)=d(\phi)+1;\\
    d(\C\phi)=d(\phi)+1.
\end{array}\]
\end{definition}

To construct two classes of models, we first define two kinds of sets of possible worlds.
\begin{definition}
For every $n\geq 1$, we inductively define two sets of possible worlds $T_n$ and $Z_n$:

\begin{itemize}
   \item  $T_0=\{t_{00}\}$ and $Z_0=\{ z_0 \}$; 
   \item  If $t_i\in T_n$, then $t_{i0}\in T_{n+1}$ and $t_{i1}\in T_{n+1}$;  if $z_i\in Z_n$, then $z_{i0}\in Z_{n+1}$ and $z_{i1}\in Z_{n+1}$;
   \item  $T_n$ and $Z_n$ have no other possible worlds.
\end{itemize}

where $|j|$ denotes the length of the subscript sequence $j$ in each $t_j$ and $z_j$.
\end{definition}

Then we define two classes of models mentioned above.
\begin{definition}
Define the class of models $\M=\{\M_n=\lr{W_n,R_n,V_n}\mid n\in\mathbb{N}^+\}$, where
\begin{itemize}
\item $W_n=T_n\cup\{r,t_0\}$,
\item $R_n=\{(t_i,t_{i0}),(t_i,t_{i1})\mid t_i\in T_n\}\cup\{(r,t_0),(t_0,t_{00})\}$,
\item $V_n(p)=W_n-\{t_{0i}\}$, where $|i|=n+1$ and $i$ is a finite sequence of $0$s.
\end{itemize}

The class of models $\N=\{\N_n=\lr{W'_n,R'_n,V'_n}\mid n\in\mathbb{N}^+\}$, where
\begin{itemize}
\item $W_n'=W_n\cup Z_n$
\item $R'_n=R_n\cup\{(z_i,z_{i0}),(z_i,z_{i1})\mid z_i\in Z_n\}\cup\{(r,z_0)\}$
\item $V_n'(p)=V_n(p)\cup Z_n-\{z_{0i}\}$, where $|i|=1$ and $i$ is a finite sequence of $0$s.
\end{itemize}
\end{definition}

It is easy to see that for any $n\in\mathbb{N}^+$, $M_n$ is a submodel of $N_n$. We will prove that {\em no} $\PLKC$-formula can distinguish $\M$ and $\N$ with the $CL$-game, which is defined below.

\begin{definition}\label{df.clgame}
A $CL$-game is a game with two players, duplicator and spoiler, playing on a Kripke-model. Given two Kripke models $M=\langle W,R,V\rangle$ and $M'=\langle W',R',V'\rangle$, from an arbitrary node $w$ in $W$ and an arbitrary node $w'$ in $W'$, play games in $n$ rounds between duplicator and spoiler as following rules:

\begin{itemize}
\item When $n=0$, if the sets of satisfied formulas on node $w$ and $w'$ are the same, then duplicator wins; otherwise, spoiler wins.
\item When $n\not=0$, 

\begin{itemize}
    \item $K$-move: If spoiler starting from node $w$ does $K$-move to node $x$ which can be reached by $R$, then duplicator starting from $w'$ does $K$-move to a node $y$ in $W'$ with the same set of satisfied propositional variables as $x$. If spoiler starts from $w'$, then duplicator starts from $w$ with similar way to move.
    \item $C$-move: If spoiler starting from node $w$ does $C$-move to node $x$ which can be reached by $\twoheadrightarrow$, then duplicator starting from $w'$ does $C$-move to a node $y$ in $W'$ with same set of satisfied propositional variables to $x$. If spoiler starts from $w'$, then duplicator starts from $w$ with similar way to move.
\end{itemize}
In the game, for arbitrary $x\in W$ and $y\in W'$, we say $(x,y)$ or $(y,x)$ is a state of $CL$-game.
\end{itemize}

\end{definition}

If there is a winning strategy for duplicator in $n$-round games, $(M_n,r)$ and $(N_n,r)$ agree on all $KCL$-formulas whose modal depth is $n$.

\begin{lemma}\label{th.nwscl}
For arbitrary $n\in \mathbb{N}$, duplicator has a winning strategy in the $CL$-game on $(M_n,r)$ and $(N_n,r)$ in $n$ rounds.
\end{lemma}

\begin{proof}
We describe duplicator’s winning strategy case by case. Starting with the initial state $(r,r)$, we mainly concerns the case where spoiler does a $K$-move. Otherwise, duplicator can move to a isomorphic sub-model such that there must be a winning strategy in following rounds. Thus, the cases below exhaust all possibilities.

\begin{itemize}
    \item The initial state is $(r,r)$:
    
    \begin{itemize}
        \item If spoiler does a $K$-move or a $C$-move on $M_n$ reaching $t_i$, then duplicator does a $K$-move or a $C$-move on $N_n$ to reach the corresponding $t_i$. Since $(M_n,t_i)\cong (N_n,t_i)$, there is a winning strategy after this move.
        \item If spoiler does a $K$-move or a $C$-move on $N_n$ reaching $t_i$, then duplicator does a $K$-move or a $C$-move on $M_n$ to reach the corresponding $t_i$. Since $(M_n,t_i)\cong (N_n,t_i)$, there is a winning strategy after this move.
        \item If spoiler does a $K$-move on $N_n$ reaching $z_0$, duplicator moves to $t_0$.
        \item If spoiler does a $C$-move on $N_n$ reaching an arbitrary node $z_{i(i\not=0)}$ in $Z_n$, then duplicator does a $C$-move to reach $t_{0i}$. Since $(M_n,t_{0i})\cong (N_n,z_i)$, there is a winning strategy after this move.
    \end{itemize}
    
    \item The current state is $(z_0,t_0)$:
    
    \begin{itemize}
        \item If spoiler does a $K$-move reaching  $z_{00}$ or $z_{01}$, then duplicator moves on $M_n$ to reach $t_{00}$.
        \item If spoiler does a $K$-move reaching $t_{00}$ on $M_n$, then duplicator moves to $z_{00}$ on $N_n$.
        \item If spoiler does a $C$-move reaching  $z_{i(i\not=0)}$, then duplicator moves to $t_{0i}$. Since $(M_n,t_{0i})\cong (N_n,z_i)$, there is a winning strategy after this move.
        \item If spoiler does a $C$-move reaching $t_{00}$ on $M_n$, then duplicator moves to $z_{00}$ on $N_n$.
        \item If spoiler does a $C$-move reaching $t_{0i(i\not=0)}$\footnote{The notation $t_{0i}$ is correct since the index for every node in $T_n$ begins with $0$.}, then duplicator moves to $z_i$ on $N_n$. Since $(M_n,t_{0i})\cong (N_n,z_i)$, there is a winning strategy after this move.
    \end{itemize}
    
    \item The current state is $(z_i,t_i)$ and $i\not=0$: this means before the game gets to this state, both players have only done $K$-moves. In the current state, there have been at most $(n-1)$ rounds. Thus, $i\leq (n-1)$ and players can do next round as follows: 
    
    \begin{itemize}
        \item If spoiler does a $K$-move reaching  $z_{i0}$ or $z_{i1}$, then duplicator does a $K$-move to reach $t_{i0}$ where $M_n,t_{i0}\models p$ since $|i0|=|i1|=(i+1)\leq n$ and there are $N_n,z_{i0}\models p$ and $N_n,z_{i1}\models p$.
        \item If spoiler does a $K$-move or a $C$-move reaching $t_{i0}$ or $t_{i1}$, then duplicator does a $K$-move or a  $C$-move to reach $z_{i0}$ where $M_n, z_{i0}\models p$ since $|i0|=|i1|=(i+1)\leq n$ and there are $N_n,t_{i0}\models p$ and $N_n,t_{i1}\models p$.
        \item If spoiler does a $C$-move reaching $z_{j(|j|>|i|)}$, then duplicator does a $C$-move to reach $t_{0j}$. Since $(M_n,t_{0j})\cong (N_n,z_j)$, there is a winning strategy after this move.
        \item If spoiler does a $C$-move reaching $t_{0i(i\not=0)}$, duplicator moves to $z_i$ on $N_n$. Since $(M_n,t_{0i})\cong (N_n,z_i)$, there is a winning strategy after this move.
    \end{itemize}
   
\end{itemize}

Therefore, for arbitrary $n\in \mathbb{N}$, there is a winning strategy for duplicator in the $n$-round $CL$-game over $(M_n,r)$ and $(N_n, r)$.
\end{proof}

For an arbitrary $\PLKC$-formula $\phi$, $\phi$ has finite modal depth $n$. By Lemma \ref{th.nwscl}, $\phi$ is satisfied both on $(M_n,r)$ and $(N_n,r)$, which means $\phi$ cannot distinguish  $(M_n,r)$ and $(N_n,r)$. This implies that we can never find a $\PLKC$-formula $\phi$ to distinguish the class of models $M$ and $N$. But we can find a $\PLKwCw_5$-formula $\Kw_i\Cw^5 p$ to distinguish them since for every  $(M_n,r)\in M$, $M_n,r\models \Kw_i\Cw_5 p$ and for every  $(N_n,r)\in N$, $N_n,r\models \neg\Kw_i\Cw_5 p$.

Therefore, following Theorem \ref{th.nkwcwlwtkcl} can be  proved.

\begin{theorem}\label{th.nkwcwlwtkcl}
Over $\mathcal{K}$, $\PLKwCw_5$ is not expressivity weaker than $\PLKC$.
\end{theorem}

This will follow that over $\mathcal{K}$, $\PLKwCw_5$ and $\PLKC$ are incomparable in expressivity. This is because $\PLKC$ is also not expressively weaker than $\PLKwCw_5$. To see this, consider two models $\M_1$ and $\M_2$, where in $\M_1$, $s_1$ can only see a $p$-world, and in $\M_2$, $s_2$ can only see a $\neg p$-world. It is straightforward to check that $\M_1,s_1\vDash\K p$ but $\M_2,s_2\nvDash\K p$, thus $\PLKC$ can distinguish between these pointed models. However, one can show that these two pointed models cannot be distinguished by $\PLKwCw_5$-formulas.

\section{$\PLKwCw_5$ over Binary Trees}\label{sec.binarytrees}

Because of the invalidity of the formula $ (Cw_5\phi\wedge Cw_5\psi)\to Cw_5(\phi\wedge \psi)$, the operator $Cw_5$ is not normal, in the sense that it  cannot be defined with some closures of accessibility relations standardly. However, an interesting observation over binary-tree models can be proved.

\begin{definition}
$(M,r)$ is a binary-tree model with root $r$ if $(M,r)$ is a tree model with root $r$ and for any node $t$ in $M$, $t$ has precisely two successors.
\end{definition}

\begin{theorem}\label{th.bt}
Consider the single-agent case. If $M,r\models Cw_5\phi$ where $(M,r)$ is a binary tree with root $r$, then on every layer of $(M,r)$, the number of the nodes where $\phi$ is satisfied is even.
\end{theorem}

In order to prove Theorem~\ref{th.bt}, we need to prove a stronger theorem:

\begin{theorem}\label{th.bts}
For an arbitrary formula $\phi$, if $M$ is a binary-tree model, then $M,v_m\models Kw_i^n\phi \ (1\leq n)$ iff the number of the $\phi$-satisfied nodes on the $(|m|+n)$th layer that $v_m$ can reach via relation $\twoheadrightarrow$ is even.
\end{theorem}

\begin{proof}

Given a binary tree $(M,v_0)$, where $v_0$ is the root, we firstly define the index of $M$ as follows: if there are nodes $v_m$, $t$, $r$ in $M$ and $v_m\to_it$, $v_m\to_ir$, then define the index of $t$ as $v_{m0}$ and the index of $r$ as $v_{m1}$. 

Let $v_m$ be an arbitrary node in $M$. Do induction on $n$:

\begin{itemize}

\item When $n=1$, 

\begin{itemize}
    \item Assume $M,v_m\models Kw_i\phi$. Since $M$ is a binary tree, there must be two nodes, $v_{m0}$ and $v_{m1}$ such that $v_m\to_iv_{m0}$ and $v_m\to_iv_{m1}$. Since  $M,v_m\models Kw_i^n\phi$, we have ($M,v_{m0}\models\phi$ and $M,v_{m1}\models\phi$) or ($M,v_{m0}\models\neg\phi$ or $M,v_{m1}\models\neg\phi$). Thus, on the $(|m|+1)$th layer, the number of nodes where $\phi$ is satisfied is 2 or 0, both of which are even.
    \item Assume the number of the nodes on the $(|m|+1)$th layer that $v_m$ can reach is even. That means there are only two possible cases: ($M,v_{m0}\models\phi$ and $M,v_{m1}\models\phi$) or ($M,v_{m0}\models\neg\phi$ or $M,v_{m1}\models\neg\phi$). Thus, we have $M,v_m\models Kw_i\phi$. 
\end{itemize}

\item Induction hypothesis: when $n=k$, $M,v_m\models Kw_i^k\phi (1\leq n)$  $\Leftrightarrow$ the number of the $\phi$-satisfied nodes on the $(|m|+k)$th layer that $v_m$ can reach via relation $\twoheadrightarrow$ is even.

\item When $n=k+1$, 

\begin{itemize}
    \item Assume $M,v_m\models Kw_i^{k+1}\phi$, which is equivalent to $M,v_m\models Kw_i^k Kw_i\phi$.
    
    Let $T$ be the set of nodes exactly consisting of all  $Kw_i\phi$-satisfied nodes on the $(|m|+k)$th layer that $v_m$ can reach via relation $\twoheadrightarrow$. By induction hypothesis, $|T|$ is even. let $|T|=2a$. Thus, among all the successors of $T$, the number of $\phi$-satisfied nodes is $2x+0y$, where $x+y=2a$. $2x+0y$ is surely an even number.  Let $S$ be a set of nodes only consisting of $\neg Kw_i\phi$-satisfied nodes on the $(|m|+k)$th layer that $v_m$ can reach via relation $\twoheadrightarrow$.  Since $M$ is a binary tree, let $|S|=2b$. For every node in $S$ has only one $\phi$-satisfied successor, among all the successors of $S$, the number of $\phi$-satisfied nodes is $2b$. Thus, the number of $\phi$-satisfied nodes on the $(|m|+k+1)$th layer is $2x+2b=2(x+b)$ which must be even.
    
    \item Assume $M,v_m\not\models Kw_i^{k+1}\phi$, which means $M,v_{m0}\models Kw_i^k\phi$ and $M,v_{m1}\models \neg Kw_i^k\phi$. By induction hypothesis, the number of the $\phi$-satisfied nodes on the $(|m|+k+1)$th layer that $v_{m0}$ can reach via relation $\twoheadrightarrow$ is even. And the number of the $\phi$-satisfied nodes on the $(|m|+k+1)$th layer that $v_{m1}$ can reach via relation $\twoheadrightarrow$ is odd. That means that the $\phi$-satisfied nodes on the $(|m|+k+1)$th layer that $v_m$ can reach via relation $\twoheadrightarrow$ is an even number plus an odd number, which equals to an odd number.
    
\end{itemize}

\end{itemize}

\end{proof}

\begin{remark}
Theorem \ref{th.bt} can be extended into a more general conclusion considering the multi-agent case: on any $G$-binary-tree model\footnote{A $G$-binary-tree model is a tree model where every node exactly has two $R_i$-successors for every $i\in G$.} $(M,r)$ where $r$ is the root, $M,r\models Cw_5\phi$ iff for any sequence of agents $s$ in $G$, on every layer of the subtree (of $(M,r)$) generated with $s$\footnote{A subtree (of some tree model $(M,r)$) generated with a sequence of agents $s$ is a subtree rooted with $r$ which only consists of all $s$-paths starting with $r$ in $(M,r)$.}, the number of the $\phi$-satisfied nodes is even.
\end{remark}

\section{Conclusion and Future work}\label{sec.conclusions}

This is a preliminary report on `commonly knowing whether'. We defined five possible notions of `commonly knowing whether' and studied how they are related to one another. On $\mathcal{S}5$-frames four of the five notions collapse. We prove the soundness and weak completeness of a  `commonly knowing whether'  logic on that class of frames. Finally, we study the expressivity of one of the proposed languages with respect to the standard common knowledge modal language on $\mathcal{K}$-frames.

There are a lot of future work to be done. For instance, the axiomatizations of $\Cw_1$ over $\mathcal{K}$ and over the class of other various frame classes, the axiomatizations and relative expressivity of other definitions for `commonly knowing whether'.

\section*{Acknowledgement}

The authors are greatly indebted to Yanjing Wang for many insightful discussions on the topics of this work and helpful comments on earlier versions of the paper. Jie Fan acknowledges the support of the project 17CZX053 of National Social Science Fundation of China. Xingchi Su was financially supported by Chinese Scholarship Council (CSC) and we wish to thank CSC for its fundings.

\bibliographystyle{plain}
\bibliography{main.bib}

\end{document}